\documentclass[copyright,creativecommons]{eptcs}
\providecommand{\event}{TERMGRAPH 2013} 
\usepackage{breakurl}             

\usepackage{latexsym}
\usepackage{graphicx}
\usepackage{color}
\usepackage{amsmath}
\usepackage{verbatim}
\usepackage{amsfonts,amssymb}
 
 \usepackage{tikz}
\usetikzlibrary{arrows,shapes}
 
\usepackage[arrow,matrix,curve]{xy}

\DeclareMathAlphabet{\mathcal}{OMS}{cmsy}{m}{n}

\numberwithin{equation}{section}

\newcommand{\FV}{{\mathit{FV}}}

\newcommand{\Varmult}{\mathit{FVmult}}

\newcommand{\Vdepth}{\mathit{Vdepth}}

\newcommand{\mainp}{\mathit{holep}}

\newcommand{\ignore}[1]{}

\newcommand{\val}{{\mathit{val}}}

\newcommand{\qed}{\hfill $\Box$}

\newtheorem{theorem}{Theorem}[section] 
{\bfseries}{\rmfamily}
{\bfseries}{\itshape}
\newtheorem{example}[theorem]{Example} 
\newtheorem{lemma}[theorem]{Lemma} 
\newtheorem{definition}[theorem] {Definition}
\newtheorem{proposition}[theorem] {Proposition}
\newtheorem{corollary}[theorem]{Corollary} 
\newtheorem{remark}[theorem]{Remark}

\newenvironment{proof}{{\it Proof.}} {\qed}
\newenvironment{proof*}{{\it Proof.}} {}

 \title{Linear Compressed Pattern Matching for Polynomial Rewriting (Extended Abstract)}
\author{Manfred Schmidt-Schauss    
\institute{Institut f\"ur Informatik,\\
Fachbereich Informatik und Mathematik,\\
Goethe-Universit{\"at},\\
Postfach 11 19 32,
D-60054 Frankfurt, Germany}
\email{schauss@ki.informatik.uni-frankfurt.de}
}
\def\titlerunning{Polynomial Rewriting of Compressed Terms}
\def\authorrunning{M. Schmidt-Schauss}

\begin{document}


\bibliographystyle{eptcs}
 

 

%
%


\maketitle
%
%
%
%
%
%

   \begin{abstract}
     This paper is an extended abstract of an analysis of  term rewriting where the terms in the rewrite rules 
     as well as the term to be rewritten are compressed by a singleton tree grammar (STG). This form of compression is more general than
      node sharing or representing terms as dags since also partial trees (contexts) can be shared in the compression. 
      In the first part efficient  but complex algorithms for detecting applicability of a rewrite rule under STG-compression 
      are constructed and analyzed.
      The second part applies these results to term rewriting sequences. 
      
      The main result for submatching is that finding a redex of a left-linear rule can be performed in polynomial time under STG-compression.
      
      The main implications  for rewriting and (single-position or parallel) rewriting steps are: (i)  under STG-compression, $n$  rewriting steps can be performed 
      in nondeterministic 
      polynomial time. 
       (ii) under STG-compression and for left-linear rewrite rules a sequence of $n$ rewriting steps can be performed 
       in polynomial  time,
       and (iii) for compressed rewrite rules where the left hand sides are either DAG-compressed or ground and STG-compressed, 
          and an STG-compressed target term, 
            $n$  rewriting steps can be performed in polynomial time.
   \end{abstract}

\section{Introduction}

An important concept in various areas of computer science like automated deduction, first order logic, term rewriting, type checking, 
are terms (ranked trees), and also terms containing variables (see e.g.  \cite{baader-nipkow:98}).
 The basic and widely used algorithms in these areas are matching, unification, term rewriting, equational deduction,  asf.  
 For example, a term $f(g(a,b), c)$ may be rewritten into $f(g(b,a), c)$ by the commutativity axiom $g(x,y) = g(y,x)$
for $g$. 
 Since implemented systems often deal with large terms, perhaps generated ones, it is of high interest to look for 
 compression mechanisms for terms, and consequently, also investigate variants of the known algorithms that 
 also perform efficiently on the compressed terms without prior decompression. 
 
 The  device of straight line programs (SLP)  for compression of  strings    
 is a general one  and allows  analyses 
 of correctness and complexity of algorithms \cite{rytter:04,lohrey-overview:12}. SLPs are polynomially equivalent to the LZ77-variant  of Lempel-Ziv compression \cite{ziv-lempel:77}. 
 SLPs are non-cyclic context free grammars (CFGs), where every nonterminal has exactly one production in the CFG, such that any nonterminal  represents exactly one string.
 Basic algorithms are the equality check of two compressed strings, which requires polynomial time \cite{Plandowski:94} (see \cite{lifshits:07} for an efficient version
 and \cite{jez-matching:2012} for a proposal of a further improvement),
 and the compressed pattern match, i.e., given two SLP-compressed strings $s,t$, the question whether $s$ is a substring of $t$ 
  can also be solved in polynomial time
 in the size of the SLPs. 
 
 A generalization of SLPs for the compression of terms are singleton tree grammars (STG)
  \cite{schmidt-schauss:05-stg,levyvillaretschauss:06a,gascon-godoy-schmidt-schauss:08},   
 a specialization of straight line context free tree grammars 
 \cite{lohreymaneth:05,busatto-lohrey-maneth:08,lohrey-maneth-schauss:09,lohrey-maneth-schmidtschauss:12}, 
 where linear SLCF tree grammars are polynomially equivalent to STGs \cite{lohrey-maneth-schauss:09,lohrey-maneth-schmidtschauss:12}. 
 Basic notions for tree grammars and tree automata can be found in  \cite{tata:97}.
 Besides using the well-known node sharing, also partial subtrees (contexts) can be shared in the compression. 
The Plandowski-Lifshits equality test of nonterminals can be generalized to STGs and requires polynomial 
time \cite{lohreymaneth:05,schmidt-schauss:05-stg} in the size of the STG.

A naive generalization of the pattern match is to find a compressed ground term in another compressed ground term, which can be solved
by translating this problem into a pattern match of  compressed preorder traversals of the terms. 
A generalization of the pattern match is the following submatching problem (also called encompassment): 
given two (STG-compressed) terms $s,t$, where $s$ may contain variables,
is there an occurrence of an instance of $s$ in $t$?  A special case is matching, where the question is whether there is a substitution $\sigma$,
such that $\sigma(s) = t$, which is shown to be in PTIME in \cite{gascon-godoy-schmidt-schauss:08,gascon-godoy-schmidt-schauss:TCL:2011}, 
including the  computation of  the (unique) compressed substitution. 

\vspace{1pt}

 In this extended abstract (of \cite{schmidt-schauss-linear-prelim:13}) we report informally on progress in finding 
 algorithms operating on STGs for answering the submatching question, and which only operate on the STGs.
 We show that if $s$ is STG-compressed and linear, then submatching can be solved in polynomial time (Theorem \ref{thm:linear-matching-polynomial}).
 If $s$ is ground and compressed or $s$ is DAG-compressed, we describe less complex algorithms that solve the submatching question
  in polynomial time 
 (Theorem \ref{theorem:ground-compressed-submatch} and Theorem \ref{theorem:DAG-and-uncompressed-submatch}).  
In the general case, we describe a non-deterministic algorithm that runs in polynomial time. The deterministic algorithm runs in 
time $O(n^{c|\Varmult(s)|})$  (Theorem \ref{theorem:general-submatch}), where $n$ is the size
of the STG and $\Varmult(s)$ the set of variables occurring more than once in $s$. This is an  exponential-time algorithm, 
but in a well-behaved parameter. 

As an application and an easy consequence of the submatching algorithms, a (single-position or parallel) deduction step on compressed terms by 
a compressed left-linear rewriting rule 
can be performed in polynomial time.  We also show that a sequence of $n$  rewrites with a STG-compressed  left-linear term rewriting system
on an STG-compressed target term can be performed in polynomial time (see Theorem \ref{theorem:main-2}).
%
Our result confirms  results on  complexity of rewrite derivations under DAG-compression  \cite{avanzini-moser:2010}, 
namely that rewrite systems with a polynomial runtime complexity can be implemented such that the algorithm requires  polynomial time.
%

\begin{example}\label{example-intro} Consider the term rewriting rule  $f(x) \to g(x,b)$, and let the term $t_1 = f(f(f(a)))$ be compressed 
as $C_1 \to f(\cdot)$, $C_2 \to C_1C_1$,   $T \to C_2(T'), T' \to f(a)$.
A single term rewriting step on the compressed term $t_1$ by the rule $f(x) \to g(x,b)$ would produce $T' \to g(a,b)$, and hence
the reduced  and decompressed term is $f(f(g(a,b)))$. 
  Other rewriting steps on the compressed term that do not decompress the term have
  to analyze the contexts. 
%
%
Let another term be $t_2 = f^{16}(a)$, compressed   
as $C_1 \to f(\cdot)$, $C_2 \to C_1C_1$, $C_3 \to C_2C_2$, $C_4 \to C_3C_3$, $C_5 \to C_4C_4$, $T \to C_5(a)$.
A term rewriting step on $T$ using $f(x) \to g(x,b)$ may rewrite  the context $f(\cdot)$  
and thus would produce $C_1 \to g(\cdot,b)$, and hence
reduces the term in one blow to $g(\ldots, (g(\ldots,b) \ldots),b)$, which is a  parallel rewriting step, see  Section  \ref{sec:rewriting}.
%
\end{example}


The structure of this extended abstract (of \cite{schmidt-schauss-linear-prelim:13}) is as follows. 
First the basic notions, in particular STGs, are introduced in Section \ref{sec:prelim}. 
An algorithm  for  linear submatching is explained in Section \ref{sec:linear-submatch}.
In Section \ref{sec:other-cases} we explain submatching for some special cases and also a general 
non-deterministic algorithm for term submatching of compressed patterns and terms.
 Finally, in Section \ref{sec:rewriting}, we illustrate the application in term rewriting  and 
argue that $n$ rewrites for a left-linear TRS can be performed in polynomial time.


\section{Preliminaries}\label{sec:prelim}

We will use standard notation for signatures, terms, positions, and substitutions (see e.g. \cite{baader-nipkow:98}).
A position is a word over positive integers. For two positions $p_1,p_2$, we write  $p_1 \le p_2$, if $p_1$ is a prefix of $p_2$, and 
$p_1 < p_2$, if $p_1$ is a proper prefix of $p_2$. We call two strings $w_1,w_2$ {\em compatible}, if $w_1$ is a prefix of $w_2$, or $w_2$ is a prefix of $w_1$.
We write $p[i]$ for the $i^{th}$ symbol of $p$, where $0$ is the start index,
 and $p[i,j]$ for the substring 
of $p$ starting at $i$ ending at $j$.  
 The set of free variables in a term $t$ is denoted as $\FV(t)$. 
 Let $\Varmult(s)$ be the set of variables occurring more than once in $s$.
 Terms without occurrences of variables are called {\em ground}.
 A term where every variable occurs at most once is called {\em linear}. A {\em context} is a term with a single hole, denoted as $[\cdot]$. 
 Sometimes it is convenient to view a linear term containing one variable as a context, where the single variable represents the hole. 
 As a generalization, a {\em multicontext} is a linear term, where the variable occurrences are also called holes. 
Let $\mainp(c)$ be the position (as a string of numbers) of a hole in a context $c$, and let the {\em hole depth} be the length of $\mainp(c)$. 
If $c = c_1[c_2]$ for contexts $c,c_1,c_2$, then  $c_1$ is a {\em prefix context} of $c$ and $c_2$ is a {\em suffix context} of $c$.
The notation $c[s]$ means the term constructed from the context $c$ by replacing the hole with $s$.
An $n$-fold iteration of a context $c$ is denoted as $c^n$; for example $c^3$ is $c[c[c]]$.
A {\em substitution} $\sigma$ is a mapping on variables, extended
homomorphically to terms by $\sigma(f(t_1,\ldots,t_n)) = f(\sigma(t_1), \ldots, \sigma(t_n))$. 

\begin{definition}
A term rewriting system (TRS) ~$R$ is a finite set of pairs $\{(l_i, r_i)~|~i = 1,\ldots,n\}$, called {\em rewrite rules}, 
  written $\{l_i \to r_i\}$, 
where we assume that for all $i:  l_i$ is not a variable, and $\FV(r_i) \subseteq \FV(l_i)$. 

A  term rewriting step by $R$ is $t \xrightarrow{R} t'$, if for some $i$: $t = c[\sigma(l_i)]$  
and 
$t' = c[\sigma(r_i)]$ for some context $c$ and some substitution $\sigma$.

\end{definition}

\subsection{Tree Grammars for Compression}

First we introduce string compression: A {\em straight line program} (SLP) is a context-free grammar that generates one word, has no cycles, 
and for every nonterminal $A$ there is exactly one production of the form $A \to A_1A_2$  or $A \to a$.

An application for SLPs is the representation of compressed positions in compressed terms.
We will use the well-known (polynomial-time) algorithms, constructions and their complexities on SLPs
like equality check of compressed strings, computing prefixes, suffixes, the common prefix (suffix) of two 
strings 
 (see \cite{rytter:04,gasieniec-karpinski-plandowski-rytter-LZ:96,Plandowski:94,Plandowski-Rytter:99,karpinski:95,lifshits:07,levy-schmidt-schauss-villaret:08}). 

We consider compression of terms using tree grammars:
\begin{definition}\label{def-STG} A {\em singleton tree grammar (STG)}  is
a 4-tuple $G = ({\cal T\cal N},{\cal C\cal N},\Sigma,{\cal R})$, where
${\cal T\cal N}$ are tree/term nonterminals
of arity $0$,
${\cal C\cal N}$ are context nonterminals of arity $1$,
and $\Sigma$ is a signature of function symbols (the terminals),
such that the sets ${\cal T\cal N}$, ${\cal C\cal N}$,  
and $\Sigma$ are finite and pairwise disjoint.
The set of nonterminals ${\cal N}$ is defined as
${\cal N} =  {\cal T\cal N} \cup  {\cal C\cal N}$.
The productions in ${\cal R}$ must be of the form:
\begin{itemize}
\item $A \to   f(A_1, \ldots, A_m)$, where $A,A_i \in {\cal T\cal N}$,
and $f \in \Sigma$ is an $m$-ary terminal symbol.
\item $A \to C_1A_2$ where $A,A_2\in {\cal T\cal N}$, and $C_1 \in {\cal C\cal N}$.
\item $C \to [\cdot]$ where $C \in {\cal C\cal N}$.
\item $C \to C_1C_2$, where $C,C_1,C_2 \in {\cal C\cal N}$.
\item  $C  \to f(A_1, \ldots,A_{i-1},[\cdot], A_{i+1},\ldots, A_m)$, where
$A_1,\ldots,A_{i-1},A_{i+1},\ldots,A_m \in {\cal T\cal N}$,
$C  \in {\cal C\cal N}$,
and $f \in \Sigma$ is an $m$-ary terminal symbol.
\item $A\to A_1$  ($\lambda$-production), where $A$ and $A_1$ are   
term nonterminals.
\end{itemize}
Let  $N_1 >_G N_2$ for two
nonterminals $N_1,N_2$, iff $(N_1 \to t) \in R$, and $N_2$ occurs in
$t$.  The STG must be non-cyclic, i.e. the transitive
closure $>^+_G$ must be irreflexive.
Furthermore, for every
nonterminal $N$ of $G$ there is exactly one production having $N$ as
left-hand side. 
Given a term $t$ with occurrences of nonterminals, the
derivation of $t$ by $G$ is an exhaustive iterated replacement of the
nonterminals by the corresponding right-hand sides. The result is denoted as
$\val_G(t)$. 
We will write $\val(t)$ when $G$ is clear from the context.
In the case of a nonterminal $N$  of $G$, we also say that $N$ (or $G$)
{\em generates} $\val_G(N)$ or {\em compresses} $\val_G(N)$. 
%
The depth of a nonterminal $N$ is the maximal number of $>_G$-steps starting from $N$, and the depth of $G$ is the maximal
depth of all its nonterminals. 
The size of an STG is the number of its productions, denoted as $|G|$.
\end{definition}

\begin{definition} Let $G$ be an STG and $V$ be a set of variables. Then $(G,V)$ is an {\em STG with variables},  
where additional production forms  are permitted:
\begin{itemize}
\item $A \to   x$, where $A  \in {\cal T\cal N}$ and $x \in V$.
\item $x\to A$  ($\lambda$-production), where $x \in V$ and $A \in {\cal T\cal N}$ .
\end{itemize}
This means that variables may be terminals or nonterminals, depending on the existing productions.
The measure  $\Vdepth(N,V)$  is defined as
 the maximal number of $>_G$-steps starting from $N$ until an element of $V$ or a terminal is reached,
 and $\Vdepth(G,V)$ the maximum.
 
 In the following we always mean STG with variables if variables are present. 
 
  An STG $G$ is called a {\em DAG}, if there are no context nonterminals. \qed
\end{definition}

The compression rate may be exponential in the best case, but not larger: The size of terms represented with an STG $G$ 
is at most $O(2^{|G|})$. 
 Note that the term depth of DAG-compressed  terms is at most the size of the DAG, whereas the term depth of STG-compressed
terms may be exponential in the size of the STG. Note also that every subterm in a DAG-compressed term is represented by a nonterminal,
 whereas in STG-compressed terms, there may be subterms that are only implicitly represented.
It is known that several computations in SLPs and STG, for example length computations,  can be done in polynomial time.
 Several forms of extensions  of STGs are well-behaved, such that even a sequence of $n$ such extensions will lead to only polynomial size growth.

\vspace{1mm}

 {\bf Compressed Matching.}
The investigation in \cite{gascon-godoy-schmidt-schauss:08} shows that (exact) term matching, also in
the fully compressed version including the computation of a compressed
substitution, is polynomial. I.e.\ given two nonterminals $S,T$, where $S$ may contain variables,   
 there is a polynomial time algorithm for answering 
the question 
whether there is some substitution $\sigma$
such that $\sigma(\val(S)) = \val(T)$, and also for computing the substitution, where the representation is 
 a list of variable-nonterminal pairs, and the nonterminals belong to an extension of the input STG.  
 
 \vspace{1mm}
 
 {\bf Compressed Submatching.}
  Given two first-order terms $s,t$, where $s$ (the pattern) may contain variables, 
   the submatching problem is to identify an instance of $s$ as a subterm of $t$.
   Submatching (also called encompassment relation) is a prerequisite for term rewriting. 
 
 \begin{definition}
The {\em compressed term submatching  problem} is: \\
Assume given a term $s$ which may contain variables,   
and a (ground)  term $t$, both compressed with an STG $G = G_S\ \cup\ G_T$, such that
$\val(T) = t$ and $\val(S) = s$ for term nonterminals $S \in G_S$, $T \in G_T$. 
The task is to
compute a (compressed) substitution $\sigma$ such that $\sigma(s)$ is a subterm of $t$;
also the (compressed) position (all positions) $p$ of the match in $t$ should be computed.
Specializations  are:{\em uncompressed} if $s$ is given as a plain term without any compression;   
{\em ground} if $s$ is ground;
{\em DAG-compressed}, if $s$ is DAG-compressed; and
{\em linear}, if $s$ is a linear term, i.e.\ every variable occurs at most once in $s$.  
%
%
\end{definition}

\begin{lemma}\label{lem:submatch-nonground-simple} Given an STG $G$, a term $s$ and a nonterminal
$T$, with  $\val_G(T) = t$, where $t$ is ground. 
If there is some substitution $\sigma$, such that $\sigma(s)$ is a subterm of  $t$, then 
there are the following possibilities: 
\begin{enumerate}
  \item\label{lem:submatch-nonground-simple-1} There is a term nonterminal $B$ of $G$  such that $\val_G(B) = \sigma(s)$.
  \item\label{lem:submatch-nonground-simple-2} There is a   production $B \to CB'$ in $G$, such that $\sigma(s) = c[\val_G(B')]$, 
       where $c$ is a nontrivial suffix context of $\val_G(C)$. 
     There are subcases for the hole position $p$ of $c$. 
    \begin{enumerate}
          \item \label{lem:submatch-nonground-simple-2a}(overlap case) $p$ is a position in $s$. 
          \item\label{lem:submatch-nonground-simple-2b}  $p = p_1p_2$, where $p_1$ is the maximal prefix of $p$ that is also a position in $s$.
               Then $s_{|p_1} = x$ is a variable. 
                The algorithms below have to distinguish the {\em subterm case} where $x$ occurs more than once in $s$ and the {\em subcontext case}
                   where  $x$ occurs exactly once in $s$.
%
                \end{enumerate} 
\end{enumerate}   
\end{lemma}


%
%

\section{Term Submatching with Linear Terms}\label{sec:linear-submatch}

  {\bf Overlaps of Linear Terms and Contexts.}
%
 %
  An important concept and technique used is periodicity of contexts. This is a generalization of periodicity of strings: 
for example the string ``bcabcabc" is periodic with period length $3$. A context $c$ is called periodic if $c = d^nd'$ for some  contexts $d,d'$
and a positive integer $n$, where $d'$ is a prefix of $d$. 
This is even generalized to multicontexts $c$ (linear terms, where the variables are the holes),
and where periodicity means that $c$ can be overlapped with itself at periodic positions without conflicts. 
%
%
%

We consider  overlapping multicontexts $c,c_1,c_2,\ldots$ and a context $d$.  
In particular special variants of overlaps have to be analyzed: Overlaps 
where the hole of $d$ is not compatible with any hole of $c$. 
The overlaps where a hole of $c$ is compatible with a hole of $d$ can be dealt with generalizing results from words (or words with character-holes).
If there are non-compatible overlaps of copies of $c$ with $d$, then only two configurations are possible:  parallel and  sequential
(see Proposition \ref{proposition:non-compatible-overlaps}   and Fig. \ref{fig:three-overlaps}), and there are no mixed configurations.  
Thus, 
periodicities in linear terms are not only possible along the hole-path of $d$ but also 
along other paths, and there are two different kinds of such periodicities: the parallel and the sequential variant.
A helpful technical result is a periodicity theorem that tells us that a multi-context $c$ is periodic, if there is a  multiple overlap 
of $h+2$ copies of $c$ where $h$ is the number of holes, and the overlap is sufficiently dense.
This will be used in the submatching algorithm for linear terms.

\begin{example}
Let $d = f(a_1, f([\cdot], a_1))$ and let $c = f(a_1,[\cdot])$. Then $c$ overlaps $d$ at position $\varepsilon$, which is a compatible overlap,
since the start as well as the hole position of $c$ is on the hole path of $d$. The overlap of  $c$ with $d$ at position $2$ (in $d$) is a non-compatible
overlap, since the hole of $c$ is at $2.2$, which is not a prefix or suffix of the hole path of $d$, which is $2.1$.  
\end{example}


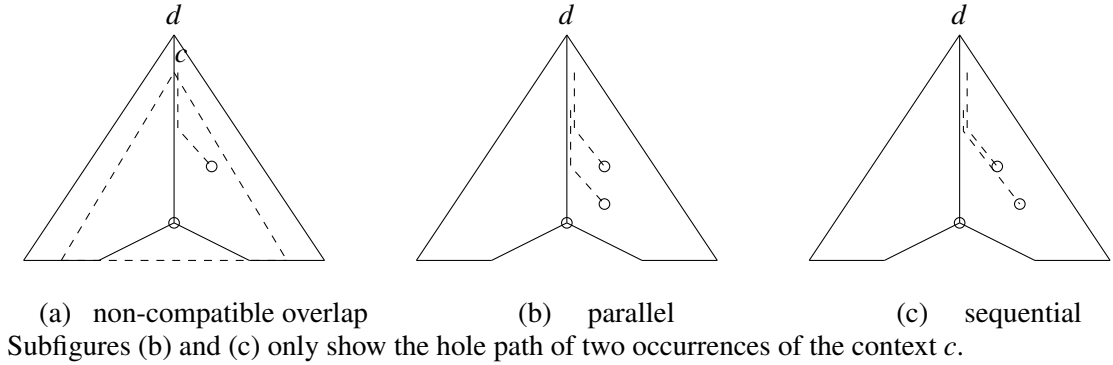
\begin{figure}[t]
\begin{tabular}{p{0.3\textwidth}p{0.3\textwidth}p{0.3\textwidth}}
\begin{minipage}{0.25\textwidth}
 \begin{tikzpicture}[auto]
     \draw (2,0.5) circle (2pt);
 \draw[-] (0,0) --   (2,3) ; \draw[-] (2,3) --  node[at start,above]{$d$}   (4,0) ; \draw[-] (0,0) --   (1,0) ;\draw[-] (1,0) --   (2,0.5) ;
 \draw[-] (2,0.5) --   (3,0) ; \draw[-] (3,0) --   (4,0) ;  
 \draw[-] (2,0.5) --   (2,3)   ;  
 \draw[dashed] (0.5,0) --   (2,2.5) ; \draw[dashed] (2,2.5) -- node[at start,above]{~~$c$}   (3.5,0) ; 
  \draw[dashed] (0.5,0) --   (3.5,0);  \draw[dashed] (2.05,2.5) --   (2.05,1.75) -- (2.5,1.25);   \draw (2.5,1.25) circle (2pt);
  \end{tikzpicture}
\end{minipage}
& 
\begin{minipage}{0.25\textwidth}
 \begin{tikzpicture}[auto]
     \draw (2,0.5) circle (2pt);
 \draw[-] (0,0) --   (2,3) ; \draw[-] (2,3) --  node[at start,above]{$d$}   (4,0) ; \draw[-] (0,0) --   (1,0) ;\draw[-] (1,0) --   (2,0.5) ;
 \draw[-] (2,0.5) --   (3,0) ; \draw[-] (3,0) --   (4,0) ;  
 \draw[-] (2,0.5) --   (2,3)   ;  
       \draw[dashed] (2.1,2.5) --     (2.1,1.75) -- (2.5,1.25);   \draw (2.5,1.25) circle (2pt);
          \draw[dashed] (2.05,2) --   (2.05,1.25) -- (2.5,0.75);   \draw (2.5,0.75) circle (2pt);
          
  \end{tikzpicture}
\end{minipage}
 &
 \begin{minipage}{0.25\textwidth}
 \begin{tikzpicture}[auto]
     \draw (2,0.5) circle (2pt);
 \draw[-] (0,0) --   (2,3) ; \draw[-] (2,3) --  node[at start,above]{$d$}   (4,0) ; \draw[-] (0,0) --   (1,0) ;\draw[-] (1,0) --   (2,0.5) ;
 \draw[-] (2,0.5) --   (3,0) ; \draw[-] (3,0) --   (4,0) ;  
 \draw[-] (2,0.5) --   (2,3)   ;  
       \draw[dashed] (2.1,2.5) --     (2.1,1.75) -- (2.5,1.25);   \draw (2.5,1.25) circle (2pt);
          \draw[dashed] (2.05,2) --   (2.05,1.72) -- (2.8,0.75);   \draw (2.8,0.75) circle (2pt);
          
  \end{tikzpicture}
\end{minipage}
\\
\quad \\
\multicolumn{1}{c}{(a)~~~non-compatible overlap} &
\multicolumn{1}{c}{(b)~~~~~parallel} &
\multicolumn{1}{c}{(c)~~~~~sequential}   
\end{tabular}
Subfigures (b) and (c) only show the hole path of two occurrences of the context $c$.
\caption{Non-compatible, parallel and sequential overlap of $c$ with $d$}\label{fig:three-overlaps}   
\end{figure}

\begin{proposition}\label{proposition:non-compatible-overlaps}  
Let $c$ be a multicontext with at least one hole, and let $d$ be a context with exactly one hole, and let $p_1 <  p_2$ be two positions
of non-compatible overlaps 
of $c$ in $d$.  Let $q_i$ be the maximal common hole path (mchp) of $c$ at $p_i$ for $i = 1,2$.
Then there are the following two cases (see Figure \ref{fig:three-overlaps}):
\begin{enumerate}
  \item $q_1 = q_2$ (the {\em parallel overlap} case).  Then for $p'$ such that $p_1p' = p_2$ the path $p_1(p')^n$ is compatible with 
  $\mainp(d)$ for all $n$.
  Also, this is a multiple overlap of $c'$   with itself at positions $(p')^i$, where $c'$ is constructed from $c$ with an extra hole at $p''$, 
   where $p_1p'' = \mainp(d)$.
   \item $q_2 < q_1$ (the {\em sequential overlap} case).    Then $p_2q_2 = p_1q_1$.   I.e., there is a fixed position on the hole path of $d$, where
   the  hole paths of occurrences of $c$ deviate. 
\end{enumerate} 
\end{proposition}

\begin{example}\label{example-overlaps}
 Let $c' = f(f(a_1,a_2),[\cdot])$ be a context, $c = f(f(x,y), (c')^{100}[.])$,
and let $d = (c')^{100}[\cdot]$. Then there is an overlap of $c$ with $d$ at positions $\varepsilon,2, 2.2, \ldots$.
It is an overlap of the first kind, i.e. a parallel overlap.
%
A sequential overlap is the following: Let $c = f(a_1, f(a_1, f(a_1,[\cdot])))$ and let 
$d = f(a_1, f(a_1, f(a_1,f([\cdot],f(a_1,f(a_1,a_1))))))$. 
Then the overlap positions are $\varepsilon, 2, 2.2, 2.2.2$.  
\end{example}
%

 \begin{theorem}[Periodicity-Theorem]\label{thm:context-overlaps-imply-periodicity}
 Let $c$ be a multi-context with $h \ge 1$ holes. Let $p$ be the position of a fixed hole of $c$, 
  and let $p_i, i = 1,\ldots,n$ be prefixes of $p$ such that $i < j$ implies $p_i < p_j$ with $n \ge h+2$.
 Assume that there is a (right-cut) overlap of $n$ copies of $c$ starting at position $p_i$ such that $p$ is a prefix of $p_ip$, i.e.,
 the hole position of $c$ starting at $p_i$ is compatible with $p$ for all $i$, and only  positions in $c$ at $p_1$ are relevant for the overlap.
Let $p_{\mathit{max}}$ be $\max\{|p_{i+1}|-|p_i|~|~i = 1,\ldots,n-1\}$. 
Assume $|p|-|p_{n}| \ge  2h\cdot p_{\mathit{max}}$;   
this means
 there are $2h\cdot{}p_{\mathit{max}}$
  common positions on the path $p$ of all occurrences of $c$. \\
Then the multicontext $c$ is periodic (in the direction $p$), and a period  length is 
$p_{\mathit{all}} : =\mathit{gcd}(\mbox{$|p_2|-|p_1|$}, \mbox{$|p_3|-|p_2|$}, \ldots, \mbox{$|p_{n}|-|p_{n-1}|)$}$. 
Moreover, the overlap is consistent with using the same substitution for the variables for every  occurrence of $c$.    
 \end{theorem}

 {\bf Tabling Prefixes of Multicontexts in Contexts.}   

The core of the algorithm for finding  submatches of a linear term $s$ in other terms (under STG-compression) 
is the construction of a table in dynamic-programming style. The table contains overlaps of $s$ with contexts that are explicitly  represented 
in the STG $G$ by a context nonterminal. 
In fact the table is split into several tables: There is a table per context nonterminal $A$ of $G$ and per variable (hole) of $s$ 
for the compatible overlaps. 
In addition there is an extra table for non-compatible overlaps. This makes $h+1$ tables where $h$ is the number of variables of $s$.

The entries in the tables are pairs of a position  and a substitution  necessary for the overlap.  
Since terms of exponential size and depth  may be represented in the STG $G$,
 a compact representation of a large number of entries is necessary in order to keep the tables of polynomial size. 
Indeed this is possible exploiting periodicity. 
If the number of entries in a table are sufficiently dense, then the periodicity theorem implies that 
a large subset of the entries enjoys regularities, and a series of periodic overlaps can be represented in one entry,
consisting of: a start position, a period (a position, respectively a context nonterminal), and the number of successive entries.

In more detail, the construction of the prefix tables is bottom-up  w.r.t. the grammar where the productions
$A \to A_1A_2$ for context nonterminals permit to construct the $A$-tables from the $A_1, A_2$-tables, and where the start are
the contexts with hole-depth 1. This construction must take into account the compact representation of the entries:
single ones and periodic ones, which makes the description of the algorithm rather complex due to lots of cases.
The construction of the prefix table in the case $A \to A_1A_2$ and the periodic cases is depicted in Figure  \ref{figure-construction-cases}
where (a) shows the case where $A$ has a periodic suffix, (b) shows the case where  $A$ has an inner part that is periodic, 
  (c) shows a case where the periodicity goes into a direction that is not compatible with the hole of $A_2$, which leads to the sequential overlap case;
and (d) is a case of a sequential overlap already in the table for $A_1$.
The generation of the periodic entries is done in an extra step: compaction, where the periodic overlaps are detected by searching 
for sufficiently dense entries. This is the only place where periodic entries are generated.

In addition to the prefix tables there is a result table, which contains the detected submatchings, and which is  
maintained during construction of the prefix tables. 

Since it is necessary to also have submatchings in terms, i.e. for term nonterminals, we keep things simple and  assume 
that every production for a term nonterminal is of the form $A \to CA_1$, where $A_1$ is a term nonterminal with production $A_1 \to a$, 
i.e. a constant. This rearrangement of $G$ can be done efficiently, and thus does not restrict generality. For these nonterminals
the extraction of the submatchings can be done using the already constructed prefix-tables.

Note that during construction of the tables, the STG $G$ may have to be extended in every step.

\begin{example}
We describe several small examples for compatible entries in a prefix table. Therefore we slightly extend Example \ref{example-overlaps}. 
Let the STG be $S \to A;  A \to A_1A_1; A_1 \to A_2A_2, A_2 \to f(a_1,[\cdot])$. 
\begin{enumerate}
\item Then $(C,A_2,\infty)$ for $C \to [\cdot]$ is a potential entry in a result table for $A$.
\item Let $A_4 \to g([\cdot]), B \to A_4A, C' \to A_4$. Then $(C',A_2,\infty)$ is an entry in the result table for $B$.
\item\label{example-pref-enum-3}  Let $B' \to BA_4$, then $(A_4,A_2,2)$ is a potential entry in the result table for $B'$.
\item\label{example-pref-enum-4}  The tuple $(A_4,A_2,3)$  is an entry in the prefix table for $B$.
\item Let $B'' \to A_6A_4, A_6 \to A_4A_1$. The context $A_6$ is then a potential  entry in the result and prefix tables of $B''$. 
\end{enumerate}
Note that item \ref{example-pref-enum-4} cannot be used as a result, since composing $B$ as in $B' \to BA_4$ in item \ref{example-pref-enum-3},
may render an overlap invalid. 
\end{example}

\begin{example}
We describe an example  for a non-compatible entry in a prefix table. Therefore we slightly modify Example \ref{example-overlaps}. 
 Assume there is an STG $G$. Let $c = f(a_1, f(a_1, f(a_1, f(a_1,[\cdot]))))$, $d = f(a_1, f(a_1, f(a_1,f([\cdot],f(a_1,f(a_1,a_1))))))$,
  and let $P,D,C_0,S$ be a nonterminals  such that $\val(P) = f(a_1,[\cdot])$, $\val(D) = d, \val(S) = c$, $\val(C_0) = [\cdot]$. 
   Then an entry in the 
  non-compatible prefix 
  table for $D$ could be $(C_0,P,3)$.
\end{example}

\begin{figure}[t]
\begin{tabular}{p{0.22\textwidth}p{0.22\textwidth}p{0.22\textwidth}p{0.3\textwidth}}
\begin{minipage}{0.22\textwidth}
 \begin{tikzpicture}[auto]
     \draw (1,0) circle (2pt);
 \draw[-] (0,0) --   (2,0) ; \draw[-] (0,0) --   (1,1) ; \draw[-] (1,1) --   (2,0) ;
 \draw[-] (0,1) --   (2,1)   ; \draw[-] (0,1) --   (1,2) ; \draw[-] (1,2) --   (2,1) ;
  \draw[-] (0,2) --   (2,2)   ; \draw[-] (0,2) --   (1,3) ; \draw[-] (1,3) --   (2,2) ;
   \draw[-] (0,3) --   (2,3)   ; \draw[-] (0,3) --node[at start,above]{$A_2~~$}     (1,4) ; \draw[-] (1,4) --   (2,3) ;
   \draw[-] (0,3.5) --   (2,3.5);
    \draw[-] (0,4) --   (2,4)   ; \draw[-] (0,4) --   (1,5) ; \draw[-] (1,5) --   (2,4) ;
     \draw[-] (0,5) -- node[midway,above]{$P$}   (2,5)   ; \draw[-] (0,5) --   (1,6) ; \draw[-] (1,6) --   (2,5) ;
      \draw[-] (0,6) -- node[midway,above]{$C$}   (2,6)   ; \draw[-] (0,6) -- node[midway,above]{$A_1$}    (1,6.5) ; \draw[-] (1,6.5) --   (2,6) ;
  \end{tikzpicture}
\end{minipage}
& 
\begin{minipage}{0.22\textwidth}
 \begin{tikzpicture}[auto]
 \draw[-] (0,0) --   (2,0.5) ; \draw[-] (0,0) --   (1,1) ; \draw[-] (1,1) --   (2,0.5) ;
 \draw[-] (0,1) --   (2,1)   ; \draw[-] (0,1) --   (1,2) ; \draw[-] (1,2) --   (2,1) ;
  \draw[-] (0,2) --   (2,2)   ; \draw[-] (0,2) --   (1,3) ; \draw[-] (1,3) --   (2,2) ;
   \draw[-] (0,3) --   (2,3)   ; \draw[-] (0,3) --node[at start,above]{$A_2~~$}     (1,4) ; \draw[-] (1,4) --   (2,3) ;
   \draw[-] (0,3.5) --   (2,3.5);
    \draw[-] (0,4) --   (2,4)   ; \draw[-] (0,4) --   (1,5) ; \draw[-] (1,5) --   (2,4) ;
     \draw[-] (0,5) -- node[midway,above]{$P$}   (2,5)   ; \draw[-] (0,5) --   (1,6) ; \draw[-] (1,6) --   (2,5) ;
      \draw[-] (0,6) -- node[midway,above]{$C$}   (2,6)   ; \draw[-] (0,6) -- node[midway,above]{$A_1$}    (1,6.5) ; \draw[-] (1,6.5) --   (2,6) ;
        \draw (1,0.25) circle (2pt);
  \end{tikzpicture}
\end{minipage}
& 
\begin{minipage}{0.22\textwidth}
 \begin{tikzpicture}[auto]
 \draw[-] (0,0) --   (2,0.5) ; \draw[-] (0,0) --   (1,1) ; \draw[-] (1,1) --   (2,0.5) ;
 \draw[-] (0,1) --   (2,1)   ; \draw[-] (0,1) --   (1,2) ; \draw[-] (1,2) --   (2,1) ;
       \draw (1.5,2) circle (2pt);
  \draw[-] (0,2) --   (2,2)   ; \draw[-] (0,2) --   (1,3) ; \draw[-] (1,3) --   (2,2) ;
   \draw[-] (0,3) --   (2,3)   ; \draw[-] (0,3) --   (1,4) ; \draw[-] (1,4) --   (2,3) ;
   \draw[-] (0,4) --   (2,4)   ; \draw[-] (0,4) --node[at start,above]{$A_2~~$}     (1,5) ; \draw[-] (1,5) --   (2,4) ;
   \draw[-] (0,4.5) --   (2,4.5);
    \draw[-] (0,5) --   (2,5)   ; \draw[-] (0,5) --   (1,6) ; \draw[-] (1,6) --   (2,5) ;
     \draw[-] (0,6) -- node[midway,above]{$P$}   (2,6)   ; \draw[-] (0,6) --   (1,7) ; \draw[-] (1,7) --   (2,6) ;
      \draw[-] (0,7) -- node[midway,above]{$C$}   (2,7)   ; \draw[-] (0,7) -- node[midway,above]{$A_1$}    (1,7.5) ; \draw[-] (1,7.5) --   (2,7) ;
  
  \end{tikzpicture}
\end{minipage} 
%
&
\begin{minipage}{0.3\textwidth}
 \begin{tikzpicture}[auto]
 \draw[-] (0,0) --   (2,0.5) ; \draw[-] (0,0) --   (1,1) ; \draw[-] (1,1) --   (2,0.5) ;
 \draw[-] (0,1) --   (2,1)   ; \draw[-] (0,1) --   (1,2) ; \draw[-] (1,2) --   (2,1) ;
  \draw[-] (0,2) --   (2,2)   ; \draw[-] (0,2) --   (1,3) ; \draw[-] (1,3) --   (2,2) ;
   \draw[-] (0,3) --   (2,3)   ; \draw[-] (0,3) --     (1,4) ; \draw[-] (1,4) --   (2,3) ;
    \draw[-] (0,4) --   (2,4)   ; \draw[-] (0,4) --   (1,5) ; \draw[-] (1,5) --   (2,4) ;
      \draw[-] (1.5,5) --node[midway,above]{$~~~A_2$}    (3,3.5) ;   \draw[-] (1.5,5) --   (2.5,3.5); \draw[-] (2.5,3.5) --   (3,3.5);  
      \draw (1.5,5) circle (2pt);      \draw (2.75,3.5) circle (2pt);
     \draw[-] (0,5) -- node[midway,above]{$P$}   (2,5)   ; \draw[-] (0,5) --   (1,6) ; \draw[-] (1,6) --   (2,5) ;
      \draw[-] (0,6) -- node[midway,above]{$C$}   (2,6)   ; \draw[-] (0,6) --node[midway,above]{$A_1$}    (1,6.5) ; \draw[-] (1,6.5) --   (2,6) ;
  
  \end{tikzpicture}
\end{minipage}  \\\\
\multicolumn{1}{c}{(a)} &
\multicolumn{1}{c}{(b)} &
\multicolumn{1}{c}{(c)} &
\multicolumn{1}{c}{(d)}  
\end{tabular}
\caption{Cases in the construction of the prefix tables for periodic entries}\label{figure-construction-cases}
\end{figure}
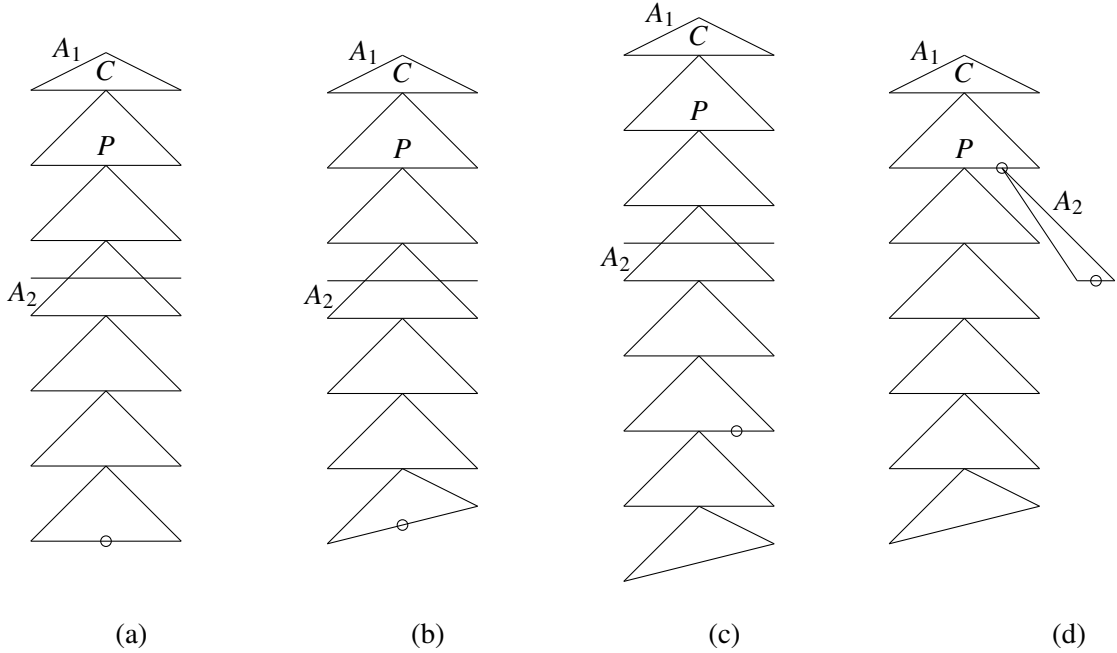

 \begin{theorem}[Linear Submatching]\label{thm:linear-matching-polynomial}
 Let $G$ be an STG, and $S,T$ be two term nonterminals such that $\val(S)$ is a linear term, and the submatching positions of $\val(S)$ in $\val(T)$
 are to be determined.
 Then the algorithm for linear submatchings  computes 
 an $O(|G|^5)$-sized representation of all submatchings of $\val(S)$
 in $\val(T)$ in polynomial time dependent on the size of $G$.
 \end{theorem}

\section{Submatching Algorithms for Other Cases}\label{sec:other-cases}

We consider  several specialized situations:  ground terms, uncompressed patterns, DAG-compressed terms,
and also non-linear terms. 

\subsection{Ground Term Submatching}

If $s$ is ground and compressed by a nonterminal $S$ then  submatching can be solved  in polynomial time by translating
both compressed terms into their compressed preorder traversals
(i.e.\ strings)   \cite{lohreymaneth:05,busatto-lohrey-maneth:08} and then applying string pattern matching 
\cite{rytter:04,lifshits:07}. 
The string matching algorithm in \cite{lifshits:07,jez-matching:2012} computes a polynomial representation of all occurrences.
Note that in our case, the structure of ground terms is very special  as a string matching problem: periodic overlaps of
the preorder traversal as strings are not possible.
Thus the complete output of the algorithm is as follows: 
(i) a list of term nonterminals $N$ of the input STG $G$, where $\val(\sigma(S)) = \val(N)$, and (ii) 
a list of pairs $(N,p)$, where the production for $N$ is of the 
form $N \to CN'$,  $p$ is a compressed position, and $\val(C)_{|\val(p)}[\val(N')] = \val(S)$. 
Moreover, every nonterminal $N$ appears at most once in the list.


The required time for string matching is $O(n^2m)$ where $n$ is the size of the SLP
of $T$ and $m$ is the size of the SLP of $S$. 
Since the preorder traversal can be computed in linear time  (see \cite{gascon-godoy-schmidt-schauss:TCL:2011}), we have: 
\begin{theorem}\label{theorem:ground-compressed-submatch}
The ground compressed term submatching can be computed in time $O(|G_T|^2|G_S|)$, and the output is a list of linear size.
\end{theorem}

\subsection{DAG-Compressed Non-Linear Submatching}\label{subsec-dag-compressed}

Now we look for the case of DAG-compressed $s$, which is slightly more general than the uncompressed case, 
and  where  variables may occur several times in $s$. Also for this case, there is an algorithm for submatching that requires polynomial time. 
The algorithm outputs enough information to determine all the positions and substitutions of a submatch.  

\begin{example}\label{ex:exponentially-many-subs}
The number of possible substitutions for a submatch in a DAG-compressed term may be exponential: 
Let the productions be $S \to f(x,y)$, and $T \to f(A_1,A_1), A_1 \to f(A_2,A_2), \ldots, A_{n-1} \to f(A_n,A_n), A_n \to a$.
Then $\val(T)$ is a complete binary tree of depth $n$  and there is a  submatch at every non-leaf node.
Clearly, it is sufficient to have all $A_i$ as submatchings in the output, which is of linear size. 
\end{example}


   
%

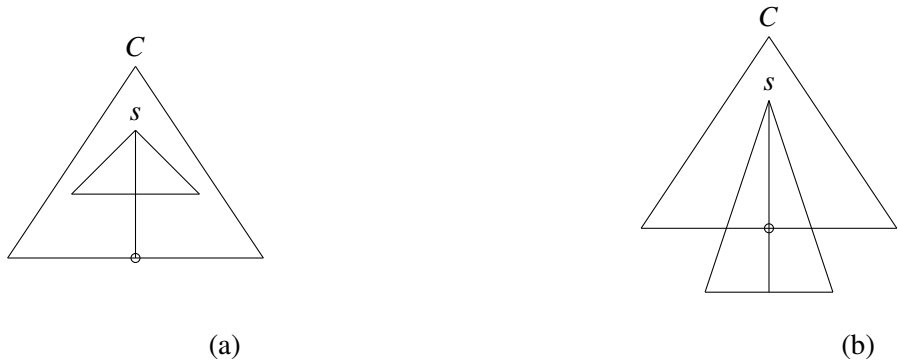
\begin{figure}[t]
\begin{tabular}{p{0.5\textwidth}p{0.5\textwidth}}
\begin{minipage}{0.5\textwidth}
 \begin{tikzpicture}[scale=0.85]  
     \draw (2,0) circle (2pt);
 \draw[-] (0,0) --   (2,3) ; \draw[-] (2,3) --  node[at start,above]{$C$}   (4,0) ; \draw[-] (0,0) --   (4,0) ;
 \draw[-] (1,1) --   (3,1)   ; \draw[-] (1,1) --   (2,2) ; \draw[-] (2,2) -- node[at start,above]{$s$}   (3,1) ; 
  \draw[-] (2,2) --   (2,0);
  \end{tikzpicture}
\end{minipage}
& 
\begin{minipage}{0.5\textwidth}
\begin{tikzpicture}[scale=0.85]  
     \draw (2,1) circle (2pt);
 \draw[-] (0,1) --   (2,4) ; \draw[-] (2,4) --  node[at start,above]{$C$}   (4,1) ; \draw[-] (0,1) --   (4,1) ;
 \draw[-] (1,0) --   (2,3)   ; \draw[-] (2,3) -- node[at start,above]{$s$}    (3,0) ; \draw[-] (1,0) --  (3,0) ; 
  \draw[-] (2,3) --   (2,0);
  \end{tikzpicture}
%
\end{minipage}   \\
\quad \\
\multicolumn{1}{c}{(a)~~~~~~~~~~~~~~~~~~~~~~~~} &
\multicolumn{1}{c}{(b)~~~~~~~~~~~~~~~~~~~~~~~~}   
\end{tabular}
\caption{Cases in  the construction of the s-in-C-table for DAG-compression}\label{fig:cins-cases}   
\end{figure}

 In the case of a DAG-compressed or uncompressed pattern-term (not necessarily linear) $s$  and STG-compressed target term $t$,
   the algorithm  for computing all submatchings is designed in dynamic programming style. It constructs a table
   of possible submatchings of $s$ in the context nonterminals corresponding to $t$. The key of the table is $(C,p)$,
   where $C$ is a context nonterminal, and $p$ a position that is a suffix of $\val(C)$ as well as a position in $s$.
    The number of these positions is linear in  $|G_s|+|G_t|$ for every context.  
   The entries are substitutions into the variables of $s$, i.e. a list of pairs $(x_i,A_i)$, where $A_i$ is a term nonterminal 
    representing a ground term.
    There is also a result list of found submatchings  in contexts $C$ contributing to $T$, 
    and term nonterminals for ground terms that are instances of $s$. 
    The construction proceeds again bottom-up in the STG $G_t$ for context nonterminals, and for 
    $A \to A_1A_2$, constructs the table for $A$ from the tables for $A_1, A_2$, and in case a full submatching is found, inserts
    a result into the result list.
    
    Finally, from these information, a representation of all submatchings can be constructed by looking at the right hand sides
    of the productions $A \to CB$ for term nonterminals, and using the table entries for $C$, and also constructing the occurrences
    of the ground terms.

\begin{theorem}\label{theorem:DAG-and-uncompressed-submatch} Let $G$ be an STG, and $S,T$ be two term nonterminals 
such that $S$ is   
DAG-compressed. Then the
submatch  computation  problem can be solved in polynomial time. Also an explicit  polynomial representation 
of all matching possibilities can be computed in polynomial time.
\end{theorem}

\subsection{A Non-Deterministic Algorithm for Sub-Matching in the General Case}\label{subsec:submatching-general}

%

The submatching problem  for STG-compressed pattern terms that may be nonlinear can be solved by a relatively easy
search that leads to a non-deterministic polynomial time algorithm: Given $S$, with non-linear $s = \val(S)$, extract and construct
  a nonterminal $B$ representing a subterm $f(r_1,\ldots,r_n)$ of $s$ such that two terms $r_i,r_j$  contain a common
  variable. 
Then  non-deterministically choose a right hand side $r$ of a production of $G_t$ of the form $f(\ldots)$, 
then compute the usual match of $B$ with  $r$ using \cite{gascon-godoy-schmidt-schauss:08} which will produce an instantiation
of at least one variable of $\val(B)$, and hence of $s$.  Then iterate this until all variables with double occurrences are instantiated. 
For the resulting  linear  term
we know how to find all matching positions.

\begin{theorem}[Nondeterministic General Submatch]\label{theorem:general-submatch} 
Let $G$ be an STG  and $S,T$ be two nonterminals of $G$ where $\val(S)$ may contain variables. Then the
algorithm for fully compressed submatching for compressed terms $s,t$ requires at most searching in  $|G|^{|\Varmult(s)|}$
alternatives for the substitution and the computation for one alternative can be done in polynomial time.  
Thus the submatching problem is in NP. 
\end{theorem}

There remains a gap in the knowledge of the complexity of the fully compressed submatching problem for terms, which for the decision problem
is between $\mathrm{PTIME}$ and $\mathrm{NP}$.

\begin{remark}\label{remark-few-variables}
 The non-linear submatching problem can be computed in polynomial time  if there are few variable occurrences ($\le |G|)$ in $s$:  
 First linearize $s$, then use the linear compressed submatch and then perform  a postprocessing checking equality
 enforced by the variables of $s$.
\end{remark}

\section{Polynomial Compressed Term Rewriting}\label{sec:rewriting}

For our compressed representation  the natural approach to rewriting is to use parallel rewriting of 
the same subterm at several positions and by the same rewriting rule.
Note, however, that the set of redexes that are rewritten in parallel will depend  on the structure of the STG $G_t$, 
and not on the structure of the rewritten term $t$. \\[1pt]
%
Let $R$ be a compressed TRS, let $t$ be a ground term with $\val_G(T) = t$, let
  $R$ be compressed by the 
STG $G_R$ as $\{L_i \to R_i~|~i = 1,\ldots,n\}$
 where $L_i,R_i$ are term nonterminals. \\
A (parallel) term rewriting step is performed as follows:\\
First select $L_i \to R_i$ as the rule. 
There is an oracle, which is one of our submatching algorithms applied to $L_i$, for finding the redex for $\val(L_i)$ 
or the set of redexes that provides the following:  
\begin{enumerate}
\item An extension $G'$ of $G$, i.e. additional nonterminals and productions.
\item A substitution $\sigma$ as a list  of pairs: $\{x_1 \mapsto A_1, \ldots, x_m \mapsto A_m\}$, where $\FV(\val(L_i)) = \{x_1,\ldots,x_m\}$, 
      $A_i$ are term nonterminals in $G'$, and $\val(A_i)$ is a subterm of $t$. It is also assumed that the instantiation is integrated in the grammar $G'$
      as productions $x_i \to A_i$ for $i = 1,\ldots,m$.
\item A term nonterminal $A$ (corresponding to $L_i$) in $G'$ which contributes  to $\val(T)$, and a compressed position $p$.
\end{enumerate}
Then the rewriting step is performed by modifying the grammar such that somewhere in the part of the grammar contributing to $t$: 
$L_i$ is replaced by $R_i$. This will also generate an extension of $G_t$ on the fly and also a copy of the STG $G_R$ is made. 

A single-position rewriting step under STG-compression is performed in a similar way.

\begin{theorem}\label{theorem:main-2}
Let $R$ be a TRS compressed with $G_R$ and $t$ be a term compressed with an STG $G$. Then a sequence of $n$ term rewriting steps
  where submatching is a non-deterministic oracle that is not counted,
 can be performed in polynomial time. The size increase
by $n$  term rewriting steps is \mbox{$\mathcal{O}\big( |G_R|^2n^7\big(|G|^2 + |G|(\log n + 2|G_R|)  + (\log n + |G_R|)^2 \big)\big)$}.
\end{theorem}
The complexity bound is
$\mathcal{O}(n^7\log^2(n))$  depending  on the number $n$ of rewrites;
 $\mathcal{O}(|G_0|^2)$  depending on the size of $G_T$;
 and  $\mathcal{O}(|G_R|^4)$  depending on the size of  $G_R$.
%
Note that the degree of the polynomial for the estimation of the worst case running time  is worse than the space bound.
The term rewriting sequence has to be constructed (+ $1$) and Plandowski equality check has to be used in every
construction step, which contributes a factor of $3$ in the exponent. But note that there are faster deterministic
tests \cite{lifshits:07,jez-matching:2012} and even faster randomized equality checks
\cite{gasieniec-karpinski-plandowski-rytter:96,berman-karpinski-2d:02,schmidt-schauss-schnitger:12}.

Single-position rewriting requires a partial decompression of the redex position (similar to the parallel), 
 which leads to an extra increase in the size of the STG, but to the same,  still polynomial, complexity.

 
 Combining the results on submatching and sequences of rewriting, we obtain the following corollaries:
 
 \begin{corollary}\label{corr-main-0}
Let $R$ be an STG-compressed  TRS and $t$ be an STG-compressed term. Then a sequence of $n$ term rewriting steps
using  the submatching algorithm in Subsection  \ref{subsec:submatching-general}
can be performed in non-deterministic polynomial time.   
\end{corollary}
\begin{proof}
This follows from Theorems \ref{theorem:main-2} and \ref{theorem:general-submatch}.
\end{proof}

 \begin{corollary}\label{corr-main-1}
Let $R$ be a left-linear STG-compressed TRS and $t$ be an STG-compressed term. Then  $n$ term rewriting steps
where the submatching algorithms  in Subsection \ref{subsec:submatching-general}
are used 
can be performed in polynomial time.  
\end{corollary}
\begin{proof}
This follows from Theorems \ref{theorem:main-2} and \ref{thm:linear-matching-polynomial}.
\end{proof}

\begin{corollary}\label{corr-main-2}
Let $R$ be a   TRS with DAG-compressed left-hand sides and STG-compressed right hand sides 
and let 
$t$ be an STG-compressed term. Then $n$ term rewriting steps  
where the submatching algorithm in Subsection \ref{subsec-dag-compressed}    
is used 
can be performed in  polynomial time in $n$. 
\end{corollary}
\begin{proof}
This follows from Theorems \ref{theorem:main-2} and \ref{theorem:DAG-and-uncompressed-submatch}.
\end{proof}

\begin{corollary}\label{corr-main-3}
Let $R$ be an  STG-compressed TRS and $t$ be an STG-compressed term, such that the 
left hand sides of every rule has at most $|G|$ occurrences of variables. 
Then  $n$ term rewriting steps (see Remark \ref{remark-few-variables}) 
can be performed in  polynomial time in $n$. 
\end{corollary}

\section{Conclusion}  

We have constructed several polynomial algorithms  for finding a submatch under STG-compression, or restrictions thereof. 
It is also shown that  $n$ rewrite steps can be performed in polynomial time under STG-compression in several cases: 
left-linear and STG-compressed TRS, DAG-compressed or ground left hand sides of rules. Also in the general case of non-linear left hand sides
$n$ rewrites can be performed non-deterministically
in polynomial time, where a search for a redex is required. This is connected  to the open problem of the exact complexity of computing
submatches 
also for non-linear  terms.

A connection  to the results in  \cite{avanzini-moser:2010} on polynomial runtime complexity is that our results  also imply that for  
TRSs with polynomial runtime complexity the (single-position and parallel) rewriting can be implemented such that $n$ rewrite steps can be performed in polynomial time.

%
 A remaining open question is whether the general STG-compressed submatching (of nonlinear terms $s$ in $t$) can be solved in polynomial time or not. 
%

\providecommand{\doi}[1]{doi:\urlalt{http://dx.doi.org/#1}{\nolinkurl{#1}}}


%

 \end{document}